\documentclass[conference]{IEEEtran}
\IEEEoverridecommandlockouts

\usepackage{amsmath,graphicx}
\usepackage{url}
\usepackage{array}
\usepackage{soul}
\usepackage{textcomp}
\usepackage{stfloats}
\usepackage{verbatim}
\usepackage{cite}
\usepackage{numprint}
\npthousandsep{\,}
\usepackage[inline]{enumitem}
\usepackage{tcolorbox}
\usepackage{float}
\usepackage[font=small,skip=0pt]{caption}
\usepackage{makecell}
\usepackage{arydshln}
\usepackage{times}
\usepackage{bm,amssymb,amsthm,amsfonts}
\usepackage{listings}%
\usepackage{algorithmicx,algorithm,algpseudocode}
\usepackage[italic]{mathastext} 
\usepackage{booktabs}
\usepackage{setspace,hyperref}
\usepackage{color}
\usepackage{etoolbox}
\usepackage{multirow}
\usepackage{pgfkeys}
\usepackage{cuted}
\usepackage{mathtools}
\usepackage[labelformat=simple]{subcaption}

\usepackage[normalem]{ulem}
\useunder{\uline}{\ul}{}
\usepackage{cancel}
\usepackage{tkz-tab}
\usepackage{latexsym}
\newtheorem{thm}{Theorem}

\newtheorem{coro}{Corollary}

\newtheorem{defi}{Definition}

\DeclareMathOperator{\vect}{vec}

\begin{document}

\title{A Hierarchical View of Structured Sparsity in Kronecker Compressive Sensing}

\author{Yanbin He and Geethu Joseph\\
Signal Processing Systems Group, Delft University of Technology, The Netherlands\\Emails: \{y.he-1,g.joseph\}@tudelft.nl}

\maketitle

\begin{abstract}
Kronecker compressed sensing refers to using Kronecker product matrices as sparsifying bases and measurement matrices in compressed sensing. This work focuses on the Kronecker compressed sensing problem, encompassing three sparsity structures: $(i)$ a standard sparsity model with arbitrarily positioned nonzero entries, $(ii)$ a hierarchical sparsity model where nonzero entries are concentrated in a few blocks, each with only a subset of nonzero entries, and $(iii)$ a Kronecker-supported sparsity model where the support vector is a Kronecker product of smaller vectors. We present a hierarchal view of Kronecker compressed sensing that explicitly reveals a multiple-level sparsity pattern. This framework allows us to utilize the Kronecker structure of dictionaries and design a two-stage sparse recovery algorithm for different sparsity models. Further, we analyze the restricted isometry property of Kronecker-structured matrices under different sparsity models. Simulations show that our algorithm offers comparable recovery performance to state-of-the-art methods while significantly reducing runtime.
\end{abstract}


\section{Introduction}
\label{sec:intro}

Multidimensional signals arise in various applications such as image processing \cite{duarte2011kronecker,caiafa2012block,caiafa2013computing,caiafa2013multidimensional} and wireless communications \cite{zhou2017low,wang2024tensor,he2023bayesian},  representing data across multiple dimensions. For example, images are intrinsically two-dimensional arrays (matrices) \cite{friedland2014compressive}, and wireless signals are dependent on the signal's angles of departure and arrival~\cite{he2022structure,he2023bayesian}. Moreover, these signals admit sparse representations in an appropriate basis, such as the discrete Fourier basis for images and angular domain in communications, enabling compressed sensing-based acquisition and recovery. Due to the physical nature of multidimensional signals, it is natural to measure the signal dimension-wise. This induces the Kronecker compressed sensing, where a Kronecker product matrix is utilized to characterize different dimensions of multidimensional signal \cite{duarte2011kronecker}. The signal model in the Kronecker compressed sensing is
\begin{equation}\label{eq.problem_basic}
    \bm y = \bm H \bm x + \bm n = \left(\bm H_1 \otimes \bm H_2\right) \bm x + \bm n,
\end{equation}
where $\bm H_1\in \mathbb{R}^{M_1\times N_1}$, $\bm H_2\in \mathbb{R}^{M_2\times N_2}$, $\bm H \in \mathbb{R}^{\bar{M} \times \bar{N}}$, $\bm x \in \mathbb{R}^{\bar{N}}$, $\bm y \in \mathbb{R}^{\bar{M}}$, and $\bm n$ is the noise, with $\bar{M}=M_1M_2$ and $\bar{N}=N_1N_2$. The goal is to recover the sparse vector $\bm x$ from noisy measurement $\bm y$ given the Kronecker product dictionary $\bm H = \bm H_1 \otimes \bm H_2$.

Furthermore, in many applications, sparse vector $\bm x$ can also exhibit additional structures. One such structure is \emph{hierarchically structured} sparsity wherein the sparse vector $\bm{x}$ in~\eqref{eq.problem_basic} is partitioned into $N_1$ blocks, each of length $N_2$. The vector $\bm{x}$ follows the $(s_1, s_2)$-hierarchical sparsity if only $s_1$ among the $N_1$ blocks are nonzero and each of these non-zero blocks is $s_2$-sparse. This structure commonly appears in channel estimation for massive multiple-input multiple-output systems \cite{roth2020reliable}. Further, if all nonzero blocks share a common support, the sparse vector is said to exhibit the $(s_1, s_2)$-Kronecker-supported sparsity. This term comes from the fact that the support of $\bm{x}$ can be expressed as the Kronecker product of two binary support vectors. This structure appears in wireless communications \cite{xu2022sparse,he2022structure,he2025kronecker} and image processing \cite{caiafa2012block,caiafa2013computing,caiafa2013multidimensional}. Motivated by the different sparsity patterns in Kronecker compressed sensing, we focus on recovering $\bm{x}$ from $\bm{y}$ in~\eqref{eq.problem_basic}, leveraging the prior knowledge of its sparsity pattern and the Kronecker matrix $\bm{H}$. Particularly, we examine three models: the standard sparsity, hierarchical sparsity, and Kronecker-supported sparsity models.

Research on sparse recovery with Kronecker compressed sensing often focuses on algorithms and guarantees. Standard sparsity models typically employ $\ell_1$-norm-based algorithms~\cite{duarte2010kronecker} or greedy algorithms~\cite{foucart2013mathematical}, that do not exploit the Kronecker structure of $\bm H$. However, some specialized algorithms have been designed for hierarchical sparsity and Kronecker-supported sparsity models. A notable example is a hard-thresholding pursuit-based algorithm for hierarchical sparsity model~\cite{roth2020reliable}. While it explicitly enforces $(s_1,s_2)$-hierarchical sparsity using a specific thresholding operator, it does not exploit the Kronecker structure of $\bm{H}$ and requires the true value of $s_1$ and $s_2$. For Kronecker-supported sparsity models, there are multiple approaches such as greedy algorithm \cite{caiafa2013computing} and Bayesian algorithms \cite{chang2021sparse,he2023bayesian}. The state-of-the-art is a sparse Bayesian learning (SBL)-based algorithm for Kronecker-supported sparse recovery~\cite{he2023bayesian}. This work explicitly enforces the Kronecker-structured support of $\bm{x}$ and leverages the Kronecker structure of $\bm{H}$ to reduce complexity \cite{he2022structure}. However, the algorithm still suffers from high complexity. Thus, we aim to design a generalized framework with more efficient algorithms by fully exploiting the Kronecker structure of $\bm H$.

Apart from algorithm development, theoretical guarantees have also been studied for the Kronecker compressed sensing, which mostly focuses on the \emph{restricted isometry property} (RIP)-based analysis~\cite{candes2008introduction}. The standard RIP analysis establishes the robust recovery of standard compressed sensing algorithms like basis pursuit and iterative hard thresholding. The RIP analysis is also extended to Kronecker-structured dictionaries for standard sparsity \cite{duarte2010kronecker,duarte2011kronecker,duarte2011kroneckertech} and hierarchical sparsity~\cite{roth2020reliable}. However, the RIP analysis of Kronecker-supported sparse vectors is missing in the literature. Further, there is no unified discussion for the RIP of Kronecker-structured matrices regarding sparse vectors with patterns. Our framework provides an integrated approach for the RIP analysis for Kronecker product matrix $\bm H$ with respect to sparse vector $\bm x$ with different patterns. 


In this paper, we present a \emph{hierarchical view} on Kronecker compressed sensing, capturing its dimension-wise signal acquisition and explicitly revealing multi-level sparsity patterns. Our contributions are twofold.
\begin{itemize}[leftmargin = *]
    \item \emph{Algorithm Design}: We design a two-stage sparse recovery algorithm that leverages the Kronecker structure of the dictionary via the hierarchical view, reducing complexity while maintaining competitive recovery performance.

    \item \emph{Theoretical Analyses}: Using the hierarchical view, we reexamine the RIP analysis for Kronecker compressed sensing, providing a generalized result that encompasses structured sparsity, with the three sparsity models as special cases.
\end{itemize}
Overall, we introduce a new perspective that bridges different sparsity models in Kronecker compressed sensing, leading to efficient algorithms and unified analysis.


\section{Hierarchical View and Two-Stage Sparse Recovery Algorithm}\label{sec.hie_mea_scheme}
This section explores the hierarchical view of Kronecker compressed sensing and a recovery approach based on it.

\subsection{Hierarchical View}\label{sec.hir_view}


Our hierarchical view relies on the structure of the measurement system in the Kronecker compressed sensing model~\eqref{eq.problem_basic}, where the Kronecker-structured matrix $\bm H$ has two factor matrices, $\bm H_1$ and $\bm H_2$. These factor matrices operate at different levels: $\bm H_1$ captures block-level while $\bm H_2$ focuses on intra-block, following a hierarchical structure. To illustrate this view, we first partition $\bm x$ into $N_1$ blocks of size $N_2$. Let the $i$th block be denoted by $\bm x_i\in\mathbb{R}^{N_2}$. Then, we can rearrange $\bm x$ into $\bm X\in \mathbb{R}^{N_2\times N_1}$ such that the $i$th column of $\bm X$ is $\bm x_i$. Similarly, we can rearrange $\bm y$ and $\bm n$ into matrices $\bm Y\in \mathbb{R}^{M_2\times M_1}$ and $\bm N\in \mathbb{R}^{M_2\times M_1}$, such that $\vect(\bm Y) = \bm y$ and $\vect(\bm N) = \bm n$, respectively. Here, $\vect(\cdot)$ denotes vectorization. Since $\vect(\bm H_2\bm X\bm H_1^\top) = \left(\bm H_1 \otimes \bm H_2\right)\bm x$, we obtain  
\begin{equation}\label{eq.step1}
    \bm Y^\top = \bm H_1 \left(\bm H_2 \bm X \right)^\top + \bm N^\top,
\end{equation}
where the $i$th row of $\bm Y$ represents the $i$th row of $\bm H_2 \bm X$ measured by $\bm H_1$. Also, the $i$th row of $\bm H_2 \bm X$ corresponds to the block $\bm x_i$ measured by $\bm H_2$. Therefore, $\bm H_1$ measures at a higher level by operating on the rows of $\bm H_2 \bm X$, effectively capturing the sparsity structure at the block level. 

The above perspective can also be interpreted directly from~\eqref{eq.problem_basic}. Recall that the Kronecker product matrix $\bm H$ possesses a column-block structure with a repetitive block pattern along its columns. Here, each block of columns is obtained by taking the Kronecker product of a column of $\bm H_1$ with $\bm H_2$. Also, the column-block structure of $\bm H$ matches with the blocks of $\bm x$. Hence, in this hierarchical framework, $\bm H_2$ first measures each block of $\bm x$. The resulting measurements of all blocks are then processed by $\bm H_1$, which captures information at a higher, global level. The relation \eqref{eq.step1} explicitly captures this measurement model, where intra-block measurement by $\bm{H}_2$ is followed by block-level measurement by $\bm{H}_1$.

\subsection{Algorithm Development}
The hierarchical view in the Kronecker compressed sensing problem~\eqref{eq.problem_basic} indicates that the sparse vector $\bm x$ can also be recovered in a hierarchical manner, leading to a two-stage recovery approach, as discussed next.

The first step of the algorithm treats $\left(\bm H_2 \bm X \right)^\top$ as unknown and estimates it by solving~\eqref{eq.step1}. Also, $\bm H_2 \bm X$ exhibits a column-wise sparsity pattern, i.e., a nonzero column of $\bm H_2 \bm X$ corresponds to a nonzero block of $\bm x$ while a zero column corresponds to a zero block. Thus, $\left(\bm H_2 \bm X \right)^\top$ is a row sparse matrix and recovering $\left(\bm H_2 \bm X \right)^\top$ from~\eqref{eq.step1} can be formulated as a multiple measurement vector (MMV) problem. It can be solved using any MMV variants of compressed sensing algorithms, such as orthogonal matching pursuit (OMP) or~SBL.

Let the estimate of $\left(\bm H_2 \bm X \right)^\top$ after the first step be $\tilde{\bm X}$. In the next step, we treat $\tilde{\bm X}$ as measurements and recover $\bm X$ from
\begin{equation}\label{eq.step2}
    \tilde{\bm X}^\top = \bm H_2 \bm X + \tilde{\bm N},
\end{equation}
where $\tilde{\bm N}$ represents noise. For standard and hierarchical sparsity models, the support of the different blocks of $\bm{x}$ (or columns of $\bm X$) are different. So, we treat problem~\eqref{eq.step2} as multiple independent single measurement vector (SMV) problems, which can be solved either sequentially or in parallel using any standard compressed sensing algorithm. Nonetheless, for the Kronecker-supported sparsity model, problem~\eqref{eq.step2} is an MMV problem because the support is common across different blocks.  The resulting algorithm, named \underline{T}wo-\underline{S}tage \underline{R}ecovery (TSR), is summarized in Algorithm~\ref{al.deco}. 
\begin{algorithm}[hpt]
\caption{Two-stage sparse recovery}
\label{al.deco}
\begin{algorithmic}[1]
\Statex \textit {\bf Input:} Measurement $\bm y$, dictionaries $\bm H_1\in \mathbb{R}^{M_1\times N_1}$, and $\bm H_2\in \mathbb{R}^{M_2\times N_2}$

\State Re-order $\bm y$ to obtain $\bm Y$
\State Solve~\eqref{eq.step1} to obtain $\tilde{\bm X}$ using any MMV algorithm
\State Solve~\eqref{eq.step2} for $\bm X$ using any recovery algorithm (use MMV variant for Kronecker-supported sparsity model)

\Statex \textit {\bf Output:} Sparse vector $\bm x = \vect (\bm X)$
\end{algorithmic}
\end{algorithm}

We further analyze the complexity of our TSR algorithm to demonstrate the benefit of exploiting the Kronecker structure of $\bm H$ via the hierarchical view. We consider TSR combined with SBL \cite{wipf2004sparse} and MMV-SBL \cite{wipf2007empirical}, and HTP \cite{blanchard2014greedy} as sparse recovery algorithms, referred to as TSSBL, TSMSBL, and TSHTP, respectively. Assume $M_1$ and $M_2$ are $\mathcal{O}(M)$, $N_1$ and $N_2$ are $\mathcal{O}(N)$, and $M<N$. Table \ref{tab.complexity} compares the time and space complexity of our algorithms with other state-of-the-art algorithms, including SBL for standard sparsity, HiHTP \cite{roth2020reliable} for $(s_1, s_2)$-hierarchical sparsity, and AM- and SVD-KroSBL \cite{he2023bayesian} for Kronecker-supported sparsity. The results against SBL-based methods indicate that TSSBL exhibits lower time and space complexity. Since there is no SBL variant for the hierarchical sparsity model, we compare our method with HiHTP. Although TSSBL has slightly higher complexity than HiHTP, our TSHTP's complexity is lower, and we can also trade-off between time and space complexity since multiple SMV in~\eqref{eq.step2} can be handled sequentially or in parallel. Also, HTP-based algorithms' complexity depends on true sparsity levels $s_1$ and $s_2$. If these values are unknown, additional iterations from inaccurate thresholding may increase complexity due to suboptimal convergence.  

\begin{table}
\centering
\scriptsize
\caption{Complexity of different algorithms in different sparse recovery problems.
}
\begin{tabular}{l|l|l}
\hline
Method     & Time Complexity & Space Complexity \\ \hline\hline
\multicolumn{3}{c}{Recovery of $s$-sparse vectors}   \\ \hline
TSSBL       & $\mathcal{O}\big(N^3M\big)$   & $\mathcal{O}(N^2)$       \\ \hline
SBL  & $\mathcal{O}\big( N^4M^2 \big)$  & $\mathcal{O}(N^4)$         \\ \hline\hline
\multicolumn{3}{c}{Recovery of $(s_1,s_2)$-hierarchically sparse vectors}   \\ \hline
TSSBL       & $\mathcal{O}\big(N^3M\big)$   & $\mathcal{O}(N^2)$       \\ \hline
TSHTP (sequential)       & $\mathcal{O}\big(MN^2+s_1M^2+s_2M^2N\big)$   & $\mathcal{O}(MN)$       \\ \hline
TSHTP (parallel)       & $\mathcal{O}\big(M^2N + (s_1+s_2)M^2\big)$   & $\mathcal{O}(MN^2)$       \\ \hline
HiHTP \cite{roth2020reliable}  & $\mathcal{O}\big(s_1s_2 M^4 + (MN)^2\big)$  & $\mathcal{O}((MN)^2)$         \\ \hline \hline
\multicolumn{3}{c}{Recovery of $(s_1,s_2)$-Kronecker-supported sparse vectors}   \\ \hline
TSMSBL       & $\mathcal{O}\big(N^2M + N^3)\big)$   & $\mathcal{O}(N^2)$       \\ \hline
KroSBL \cite{he2023bayesian} & $\mathcal{O}\big(N^{3} +(MN)^2\big)$  & $\mathcal{O}((MN)^2)$         \\ \hline
\end{tabular}
\label{tab.complexity}
\end{table}

\section{Unified RIP Analysis For Structured Sparsity Models}

Our two-stage recovery approach suggests that the key factor for recovery is not the sparsity level $\|\bm x\|_0$, but the maximum sparsity level of different blocks, $\|\bm x_i\|_0$, and the number of nonzero blocks. We can leverage this formulation to unify the analysis for three sparsity models: standard, hierarchical, and Kronecker-supported sparsity. For this, we first introduce a generalized notion of RIP called the $\mathcal{S}$-RIP condition, where $\mathcal{S}$ is the set of sparse vectors under a given sparsity model. 

\begin{defi}
    A matrix $\bm H$ satisfies the $\mathcal{S}$-RIP, if there exists a constant $\delta\in (0,1)$ such that 
    \begin{equation}\label{def.set_RIC}
        (1-\delta) \|\bm x\|_2^2 \leq \|\bm H\bm x\|_2^2 \leq (1+\delta)\|\bm x\|_2^2,
    \end{equation}
holds for all vectors $\bm x\in \mathcal{S}$. The smallest feasible $\delta$, denoted as $\delta_\mathcal{S}(\bm H)$, is the $\mathcal{S}$-RIC of $\bm H$.
\end{defi}
Under our models, $\mathcal{S}$ is a union of subspaces. So, our $\mathcal{S}$-RIP is related to bi-Lipschitz condition in \cite{blumensath2011sampling} and can guarantee the success of compressed sensing algorithm, such as iterative hard thresholding \cite{blumensath2011sampling}. We skip the details here, but see \cite{eldar2009robust} for a discussion RIP-based conditions and structured sparsity.

By changing $\mathcal{S}$, we derive the different sparsity models. For example, if $\mathcal{S}$ is the set for all $s$-sparse vectors, it reduces to the standard $s$-RIP condition. We denote the standard RIC of a given matrix $\bm H$ as $\delta_{s}(\bm H)$. The next result presents an upper bound for $\delta_\mathcal{S}(\bm H)$, based on the standard RICs of $\bm H_1$ and $\bm H_2$.

\begin{thm}\label{thm.basic}
    For the Kronecker-structured dictionary $\bm H = \bm H_1 \otimes \bm H_2$ and set $\mathcal{S}\subseteq\mathbb{R}^{\bar{N}}$, the $\mathcal{S}$-RIC of $\bm H$ satisfies 
    \begin{equation}
        \delta_\mathcal{S}(\bm H) \leq \sup_{\bm x \in \mathcal{S}} (1+\delta_{s_1(\bm x)}(\bm H_1))(1+\delta_{s_2(\bm x)}(\bm H_2))-1,\notag
    \end{equation}
    where for any vector $\bm x\in\mathcal{S}$, the term $s_1(\bm x)$ is the number of nonzero blocks in $\bm x$ when we partition into $N_1$ blocks with each length $N_2$, and $s_2(\bm x)=\max_{i} \|\bm x_i\|_0$ where $\bm x_i$ represents the $i$th block of $\bm x$.    
\end{thm}

\begin{proof}
    
For any $\bm x\in\mathcal{S}$, we note that \eqref{def.set_RIC} bounds $\|\bm H\bm x\|_2^2=\|\bm H_1\left(\bm H_2\bm X\right)^\top\|_\mathrm{F}^2$, 
where $\bm x = \vect (\bm X)$ and the $i$th column of $\bm X$ is $\bm x_i$. We look into $\bm H_2\bm X$ first. Since there are only $s_1(\bm x)$ among $N_1$ blocks of $\bm x$ are nonzeros, the matrix $\bm H_2\bm X$ has at most $s_1(\bm x)$ nonzero columns. Hence, every column of $\left(\bm H_2\bm X\right)^\top$ has at most $s_1(\bm x)$ non-zero entries. Using the standard RIC of $\bm H_1$, 
\begin{equation}\label{eq.ric1}
    (1-\delta_{s_1(\bm x)}) \|\bm H_2\bm X\|_\mathrm{F}^2 \leq \|\bm H_1\left(\bm H_2\bm X\right)^\top\|_\mathrm{F}^2 \leq (1+\delta_{s_1}) \|\bm H_2\bm X\|_\mathrm{F}^2.
\end{equation}
Further, since there are at most $s_2(\bm x)$ non-zero entries in each column of $\bm X$, we derive
\begin{equation}\label{eq.ric2}
    (1-\delta_{s_2(\bm x)})\|\bm X\|_\mathrm{F}^2 \leq \|\bm H_2\bm X\|_\mathrm{F}^2 \leq (1+\delta_{s_2(\bm x)})\|\bm X\|_\mathrm{F}^2.
\end{equation}
Combining~\eqref{eq.ric1} and~\eqref{eq.ric2}, we conclude
\begin{multline*}
    (1-\delta_{s_1(\bm x)}) (1-\delta_{s_2(\bm x)})\|\bm X\|_\mathrm{F}^2 \leq \|\bm H_1\left(\bm H_2\bm X\right)^\top\|_\mathrm{F}^2 \\\leq (1+\delta_{s_1(\bm x)}) (1+\delta_{s_2(\bm x)})\|\bm X\|_\mathrm{F}^2.
\end{multline*}
Since $\|\bm H\bm x\|_2^2=\|\bm H_1\left(\bm H_2\bm X\right)^\top\|_\mathrm{F}^2$ and $\|\bm x\|_2^2=\|\bm X\|_\mathrm{F}^2$, we get 
\begin{align*}\label{eq.sup_max}
    \delta_{\mathcal{S}}(\bm H)&\leq \sup_{\bm x \in \mathcal{S}}\max \{1-(1-\delta_{s_1(\bm x)}) (1-\delta_{s_2(\bm x)}),\\
    &\hspace{2.5cm}(1+\delta_{s_1(\bm x)}) (1+\delta_{s_2(\bm x)})-1\}\\
    &=\sup_{\bm x \in \mathcal{S}}(1+\delta_{s_1(\bm x)}) (1+\delta_{s_2(\bm x)})-1,
\end{align*}
    which completes the proof.   
\end{proof}
We next derive the RICs for the standard $s$-sparsity, $(s_1,s_2)$-hierarchical sparsity, and $(s_1,s_2)$-Kronecker-supported sparsity by changing the definitions of $\mathcal{S}$ in Theorem \ref{thm.basic} and compare them with the existing results. We start with the standard sparsity model where $\mathcal{S}$ is the set of $s$-sparse vectors.

\begin{coro}\label{coro.s-sparse}
For the sparsity level $s$, the $s$-RIC of a Kronecker-structured matrix $\bm H = \bm H_1 \otimes \bm H_2$ can be bounded as
    \begin{equation*}
        \delta_s(\bm H) \leq \max_{1\leq s_1\leq s} (1+\delta_{s_1}(\bm H_1))(1+\delta_{s+1-s_1}(\bm H_2))-1. 
    \end{equation*}
\end{coro}

\begin{proof}
    We derive the result from Theorem \ref{thm.basic} by setting $\mathcal{S}$ as the set of $s$-sparse vectors. Using the notation used in Theorem \ref{thm.basic}, for any $s$-sparse vector $\bm x$, if $\max_{i} \|\bm x_i\|_0=s_2(\bm x)$, then there are at most $s-s_2(\bm x)+1$ nonzero blocks in $\bm x$. We hence have $s_1(\bm x) \leq s-s_2(\bm x)+1$. Thus, $\delta_s(\bm H)$ can be bound by the maximum value of $(1+\delta_{s_1}(\bm H_1))(1+\delta_{s_2}(\bm H_2))-1$ over all possible options of $s_1$ and $s_2$ satisfying the constraint, $s_1(\bm x) + s_2(\bm x) \leq s+1$. Consequently, we arrive at  
    the desired result.
\end{proof}
Our bound in Corollary \ref{coro.s-sparse} is no worse than the existing bound for $\delta_s(\bm H)$ in the literature~\cite{duarte2011kroneckertech} because
\begin{align}
\max_{1\leq s_1\leq s} (1+\delta_{s_1}(\bm H_1))(1+\delta_{s+1-s_1}(\bm H_2))-1\notag\\
    &\hspace{-4.5cm} \leq \max_{1\leq s_1\leq s} (1+\delta_{s_1}(\bm H_1))\max_{1\leq s_1\leq s}(1+\delta_{s+1-s_1}(\bm H_2))-1\notag\\
    &\hspace{-4.5cm}\leq (1+\delta_s(\bm H_1))(1+\delta_s(\bm H_2))-1,\label{eq.exist}
\end{align}
 since $\delta_s(\cdot)$ is a non-decreasing function of $s$ and \eqref{eq.exist} is the existing bound.
Further, if $\bm H_1=\bm H_2$, we have
\begin{align*}
\max_{1\leq s_1\leq s} (1+\delta_{s_1}(\bm H_1))(1+\delta_{s+1-s_1}(\bm H_1))-1\\
    &\hspace{-4.5cm} = \max_{1\leq s_1\leq \lceil s/2\rceil} (1+\delta_{s_1}(\bm H_1))(1+\delta_{s+1-s_1}(\bm H_1))-1\\
    &\hspace{-4.5cm} \leq (1+\delta_{\lceil s/2\rceil}(\bm H_1))(1+\delta_{s}(\bm H_1))-1,
\end{align*}
making our bound better than the existing bound.

Corollary \ref{coro.s-sparse} also corroborates that that only the number of nonzero blocks and the maximum number of nonzeros in each block $\max_{i} \|\bm x_i\|_0$ affect the $s$-RIC of Kronecker-structured $\bm H$.
The intra-block measuring of $\bm H_2$ and the block-level measuring of $\bm H_1$ collectively contribute to the deviation $\left|\|\left(\bm H_1 \otimes \bm H_2\right)\bm x\|_2^2-\|\bm x\|_2^2\right|$, reflected by the coefficients $\delta_{s_1}$ and $\delta_{s_2}$ in Corollary~\ref{coro.s-sparse}.

Next, we look at the hierarchical sparsity model, by choosing $\mathcal{S}$ as the set of all $(s_1,s_2)$-hierarchically sparse vectors.
\begin{coro}
   Consider the Kronecker-structured dictionary $\bm H = \bm H_1 \otimes \bm H_2$. For $(s_1,s_2)$-hierarchically sparse vectors, the $(s_1,s_2)$-RIC of $\bm H$, i.e., $\delta_{(s_1,s_2)}(\bm H)$, can be bounded as
    \begin{equation}
        \delta_{(s_1,s_2)}(\bm H) \leq (1+\delta_{s_1}(\bm H_1))(1+\delta_{s_2}(\bm H_2))-1,\notag
    \end{equation}
\end{coro}
We note that this result is identical to the result in \cite{roth2020reliable}. Further, the RIC of Kronecker-structured $\bm H$ for $(s_1,s_2)$-Kronecker-supported sparse vector $\bm x$ is as follows.
\begin{coro}
    Consider the Kronecker-structured dictionary $\bm H = \bm H_1 \otimes \bm H_2$. For Kronecker-supported sparse vector $\bm x$ where there are $s_1$ nonzero blocks with each $s_2$-sparse, the Kronecker-supported-RIC of $\bm H$ can be bounded as
    \begin{equation}
        \delta_{s_1,s_2}^{\mathrm{KS}}(\bm H) \leq (1+\delta_{s_1}(\bm H_1))(1+\delta_{s_2}(\bm H_2))-1.\notag
    \end{equation}
\end{coro}

Surprisingly, the upper bound of $\delta_{s_1,s_2}^{\mathrm{KS}}$ matches $\delta_{(s_1,s_2)}$, despite the Kronecker-supported sparse vector $\bm x$ being more ``structured" than the $(s_1,s_2)$-hierarchically sparse vector. To clarify, we recall the additional structure in the Kronecker-supported sparsity model arises from the joint sparsity of the nonzero blocks of $\bm x$, aligning \eqref{eq.step2} with the MMV model. However, for $\bm{H}_2\bm{X}$, the $(s_1,s_2)$-hierarchically sparse vector does not necessarily maintain the same support across nonzero blocks, making \eqref{eq.step2} a collection of SMV problems. As noted in \cite{li2013structured,eldar2009robust}, RIP analysis considers worst-case performance and does not guarantee that the standard MMV measurement model outperforms the SMV case. Consequently, our bound for $\delta^{\mathrm{KS}}$ offers no improvement, and to the best of our knowledge, establishing a stronger RIP-based condition for the MMV model remains an open problem in the literature.

\begin{figure*}
\centering
  \subcaptionbox{Regular $s$-sparse vector.\label{fig.s}}{\includegraphics[width=0.32\linewidth]{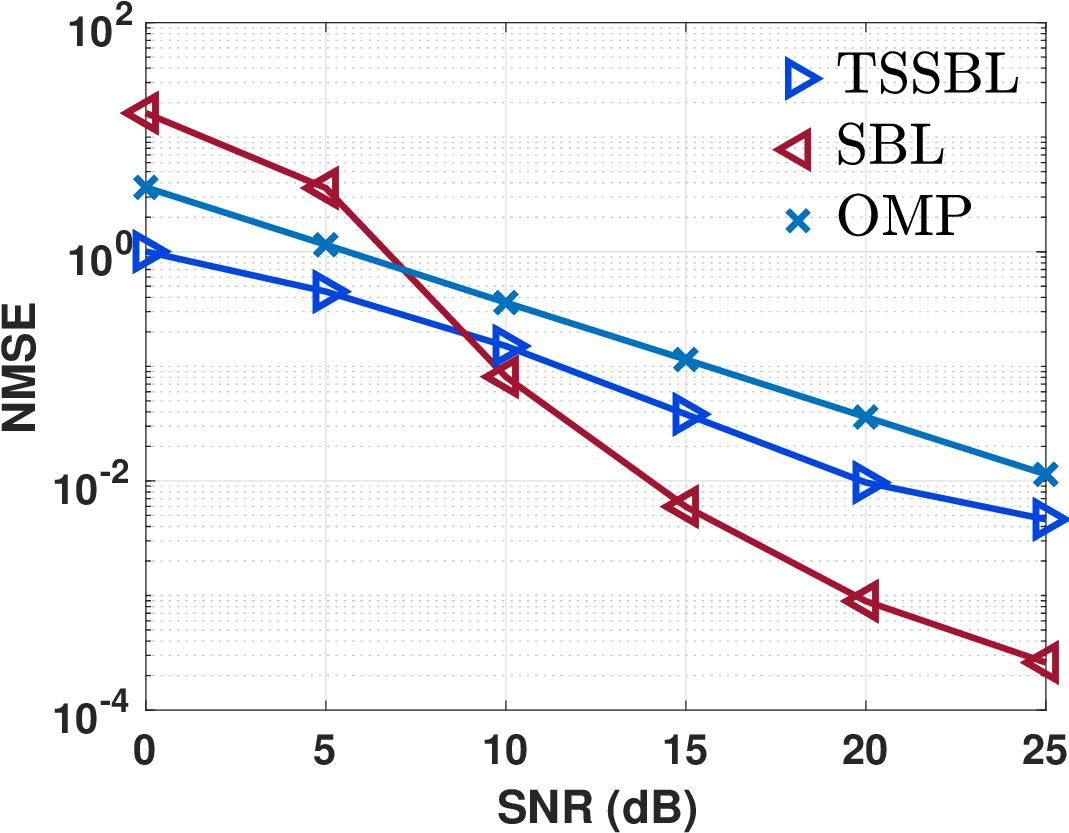}}
  \subcaptionbox{Hierarchically sparse vector.\label{fig.hi}}{\includegraphics[width=0.32\linewidth]{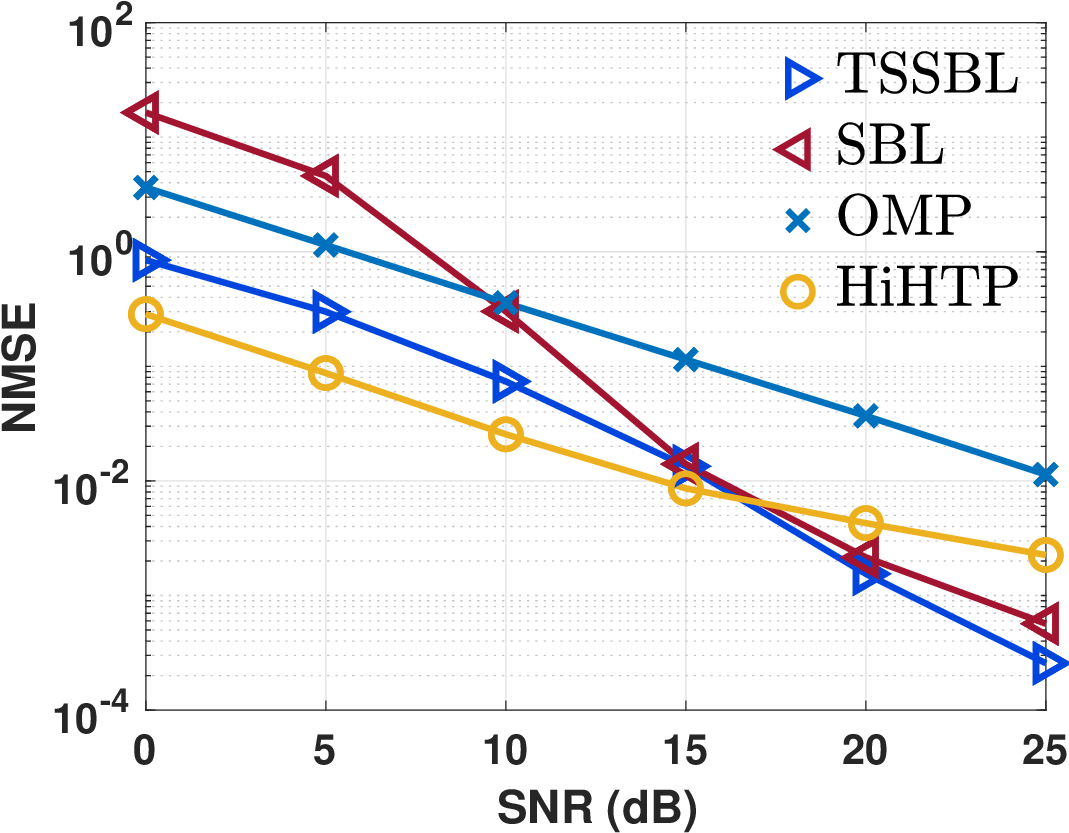}}
  \subcaptionbox{Kronecker-supported sparse vector.\label{fig.kro}}{\includegraphics[width=0.32\linewidth]{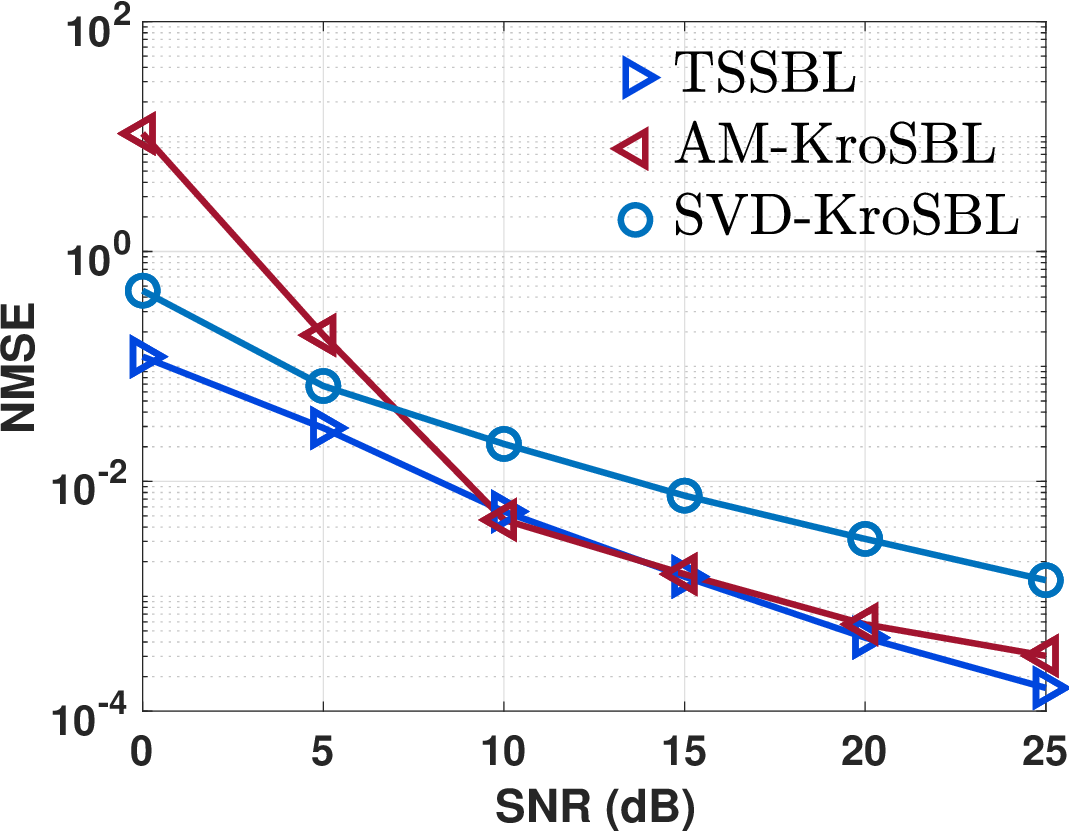}}
  \caption{NMSE performance of different algorithms as functions of SNR in different types of sparse vector recovery setting.}
  \label{fig.snr}
\end{figure*}

\section{Numerical Evaluations}

In this section, we present numerical results for the recovery of regular $s$-sparse vector, $(s_1,s_2)$-hierarchically sparse vector, and $(s_1,s_2)$-Kronecker-supported sparse vector. We use SBL and MMV-SBL \cite{wipf2007empirical} for the compressed sensing algorithm in Algorithm \ref{al.deco}. Our benchmarks for the standard sparsity model are SBL~\cite{wipf2004sparse} (state-of-the-art) and OMP. We use HiHTP \cite{roth2020reliable} (state-of-the-art), SBL, and OMP as benchmarks for the hierarchically sparse vector recovery. For $(s_1,s_2)$-Kronecker-structured support sparse vector, we benchmark with AM-KroSBL and SVD-KroSBL (state-of-the-art)~\cite{he2023bayesian}. 

Our setting is as follows. For the Kronecker-structured dictionary $\bm H=\bm H_1 \otimes \bm H_2$, we set $M_1=M_2=30$ and $N_1=N_2=40$. The entries of $\bm H_1$ and $\bm H_2$ and nonzero entries of $\bm x$ are drawn independently from the standard normal distribution. For $s$-sparse vectors, we set $s=15$, and the support is randomly drawn from a uniform distribution. For $(s_1,s_2)$-hierarchically and $(s_1,s_2)$-Kronecker-supported sparse vectors, we opt for $s_1 = s_2 = 5$. Here, supports are generated by first selecting $s_1$ blocks uniformly at random, then assigning support within each block uniformly. In the hierarchical sparsity model, support varies across blocks, whereas in the Kronecker-supported sparsity model, it remains identical. We adopt the additive white Gaussian noise with zero mean whose variance is determined by $\text{SNR~(dB)} = 10\log_{10}\mathbb{E}\{\|\bm H\bm x\|_2^2/\|\bm n\|_2^2\}$ of $\{0,5,10,15,20,25\}$. Two metrics are considered for performance evaluation: normalized mean squared error (NMSE) and run time. Here, we define 
$\text{NMSE} = \mathbb{E}\left\{{\|\bm x-\hat{\bm x}\|_2^2}/{\|\bm x\|_2^2}\right\},$
where ${\bm x}$ is the ground truth and $\hat{\bm x}$ is the estimated vector. We limit the number of iterations for all SBL-based methods (TSSBL, SBL, MMV-SBL, AM-KroSBL, and SVD-KroSBL) to three hundred. We also prune small entries in hyperparameters for faster convergence for SBL-based algorithms. Simulation results, summarized in Fig. \ref{fig.snr} and Table \ref{tab.time}, are averaged over 500 independent trials.

Fig. \ref{fig.snr} compares the performance of different algorithms under each sparsity model. In Fig. \ref{fig.s}, we compare our TSSBL with SBL and OMP. Our TSSBL outperforms OMP in all cases and outperforms SBL when SNR is low. The worse performance of TSSBL in high SNR cases is because the algorithm can ignore some blocks of $\bm x$ having fewer nonzero entries in the first step when the noise is present. The worse performance of SBL in low SNR case is because the noise estimation step in SBL is not robust in strong noise case \cite{zhang2011sparse} and can lead to degraded performance. We then compare our TSSBL with SBL, OMP, and HiHTP in Fig. \ref{fig.hi}. We note that HiHTP requires the true sparsity levels $s_1$ and $s_2$ as input, which are not always known in real applications. Thus, for a fair comparison, we do not input the true sparsity level but only a rough approximation. Due to the thresholding operator, HiHTP can effectively enforce the hierarchical sparsity, making the algorithm more robust against strong noise. However, when SNR increases, the drawback of requiring accurate sparsity level prevails, hence our TSSBL dominates. The comparison between our TSSBL and AM-/SVD-KroSBL is shown in Fig. \ref{fig.kro}, where it shows that our TSSBL is able to perform similarly or better for all SNR values. 

Furthermore, from Table \ref{tab.time}, we see that the run time our TSSBL is at least one order less than the other candidates. We attribute this to exploiting the Kronecker structure of $\bm H$. Overall, our TSSBL offers comparable performance with a considerably shorter run time than the state-of-the-art methods.

\begin{table}[t]
\centering
\scriptsize
\caption{Runtime in seconds. \textbf{Bold}: the best result.}
\begin{tabular}{l|c|c|c|c|c|c}
\hline
SNR & 0 dB     & 5 dB   & 10 dB   & 15 dB   & 20 dB   & 25 dB   \\ \hline
\hline
\multicolumn{7}{c}{Recovery of $s$-sparse vectors}   \\ \hline
TSSBL  & \textbf{0.122} & \textbf{0.082} & \textbf{0.040} & \textbf{0.022} & \textbf{0.015} & \textbf{0.010} \\ \hline
OMP  & 0.606	& 0.604	& 0.598	& 0.602	& 0.600	& 0.600 \\ \hline
SBL  & 3.676  & 4.174 & 2.526 & 1.233 & 0.986 & 0.923 \\ \hline
\hline
\multicolumn{7}{c}{Recovery of $(s_1,s_2)$-hierarchically sparse vectors}   \\ \hline
TSSBL  & \textbf{0.092} & \textbf{0.046} & \textbf{0.015} & \textbf{0.009} & \textbf{0.006} & \textbf{0.004} \\ \hline
OMP  & 0.606&	0.606&	0.599& 0.599&	0.598&	0.598 \\ \hline
SBL  & 3.559&	3.814&	3.299&	1.387&	0.977&	0.873 \\ \hline
HiHTP  & 0.821  & 0.830 & 0.840 & 0.812 & 0.826 & 0.808 \\ \hline
\hline
\multicolumn{7}{c}{Recovery of $(s_1,s_2)$-Kronecker-supported sparse vectors}   \\ \hline
TSSBL  & \textbf{0.028} &	\textbf{0.015} & \textbf{0.004} &	\textbf{0.002} &	\textbf{0.001} &	\textbf{0.001} \\ \hline
AM-KroSBL  & 7.572 & 8.145 &	3.990 &	2.142 &	1.263 &	0.815  \\ \hline
SVD-KroSBL  & 4.299 & 1.386 & 0.530 &	0.356 &	0.308 &	0.292 \\ \hline
\end{tabular}
\label{tab.time}
\end{table}

\section{Conclusion}

This work focused on the Kronecker compressed sensing problem with multiple sparsity structures. We first explored the hierarchical view of the Kronecker-structured dictionary. Each factor matrix in the Kronecker-structured dictionary measures the sparse signal at different levels. Based on the hierarchical view, we developed a two-stage sparse recovery algorithm, which offers comparable performance compared with other state-of-the-art algorithms with lower computational complexity. We then unified the RIP analysis of Kronecker-structured matrix with different structured sparsity models. Designing new algorithms through hierarchical view and establishing recovery guarantees are exciting avenues for future research. 

\bibliographystyle{ieeetr}
\bibliography{refs}

\end{document}